\documentclass[conference]{ieeeconf}  

\IEEEoverridecommandlockouts
\usepackage{graphicx} 
\usepackage{amsmath} 
\usepackage{array}   
\usepackage{booktabs} 
\usepackage{caption} 
\usepackage[normalem]{ulem}
\usepackage{siunitx} 
\usepackage{epstopdf} 
\usepackage{cite}
\usepackage{subcaption}
\usepackage{amsmath,amssymb,amsfonts,mathtools,nccmath,bm, tikz, enumerate}
\usepackage{multicol}
\usepackage[nodisplayskipstretch]{setspace}
\newtheorem{assumption}{Assumption}[section]
\newtheorem{remark}{Remark}[section]
\newtheorem{theorem}{Theorem}
\newtheorem{lemma}{Lemma} 
\newtheorem{proposition}{Proposition} 
\usepackage{graphicx}
\usepackage{textcomp}
\usepackage{cuted}
\usepackage{soul}
\usepackage{float}
\usepackage{xcolor}
\usepackage{amsfonts}
\usepackage{algorithm}
\usepackage{algpseudocode}
\usepackage[T1]{fontenc}
\def\BibTeX{{\rm B\kern-.05em{\sc i\kern-.025em b}\kern-.08em
    T\kern-.1667em\lower.7ex\hbox{E}\kern-.125emX}}
\allowdisplaybreaks

\usepackage{tabstackengine}
\usetikzlibrary{shapes,shapes.geometric,arrows.meta,fit,calc,positioning,automata,patterns}
\usetikzlibrary{tikzmark,calc,decorations.pathreplacing}
\tikzset{
    node_style/.style={
        draw, fill=green!60, text=black, regular polygon, regular polygon sides=8,
        minimum size=0.4cm, font=\scriptsize\bfseries
    },
    edge_style/.style={draw, ->, thick, purple},
    edge_style_1/.style={draw, <->, thick, purple},
}

\begin{document}


\title{\textbf{
Tube-Based Safety for Anticipative Tracking in Multi-Agent Systems}}

\author{Armel Koulong, \IEEEmembership{Student Member, IEEE} and Ali Pakniyat, \IEEEmembership{Member, IEEE}
\thanks{A. Koulong and A. Pakniyat are with the department of \mbox{Mechanical} Engineering, University of Alabama, Tuscaloosa, AL, USA (e-mails: akoulongzoyem@crimson.ua.edu, apakniyat@ua.edu). }
}

\maketitle
\thispagestyle{plain}
\pagestyle{plain}

\begin{abstract}
A tube-based safety framework is presented for robust anticipative tracking in nonlinear Brunovsky multi-agent systems subject to bounded disturbances. The architecture establishes robust safety certificates for a feedforward-augmented ancillary control policy. By rendering the state-deviation dynamics independent of the agents' internal nonlinearities, the formulation strictly circumvents the restrictive Lipschitz-bound feasibility conditions otherwise required for robust stabilization. Consequently, this structure admits an explicit, closed-form robust positively invariant (RPI) tube radius that systematically attenuates the exponential control barrier function (eCBF) tightening margins, thereby mitigating constraint conservatism while preserving formal forward invariance. Within the distributed model predictive control (MPC) layer, mapping the local tube radii through the communication graph yields a closed-form global formation error bound formulated via the minimum singular value of the augmented Laplacian. Robust inter-agent safety is enforced with minimal communication overhead, requiring only a single scalar broadcast per neighbor at initialization. Numerical simulations confirm the framework's efficacy in safely navigating heterogeneous formations through cluttered environments.
\end{abstract}



\section{Introduction} \label{Introduction}
Leader--follower formations are central to cooperative
autonomy---UAV swarms~\cite{Oh2015UAVSurvey}, mobile sensor
networks~\cite{Olfati-Saber2007ConsensusCoopControl}, and robotic
teams~\cite{Chen2016FormationRoboticsSurvey}---where agents must track a maneuvering
leader while maintaining prescribed geometry, avoiding collisions and
obstacles, and respecting hard constraints under disturbances.
Position-only sensing of the leader introduces tracking phase lag during
fast maneuvers~\cite{LiRenXuTAC2012_PositionOnlyContainment, LiRenXuACC2011_PositionOnlyTracking}.  Collision
and obstacle constraints have high relative degree for high-order
dynamics and must be enforced robustly under bounded disturbances.

Consensus and leader--follower tracking over fixed graphs are well
understood~\cite{RenBeardMcLain2003,Kan2016DirectedRandomGraphs} but typically
omit hard constraints or safety certificates.  Distributed
MPC~\cite{DunbarMurray2006,Keviczky2008, FARINA20121088, Fu70088, Shorinwa2024} addresses
coordination with coupled constraints, including collision
avoidance~\cite{Dai2017DMPCFormationSurveyLike}, but safety is often encoded via
penalties rather than forward-invariance certificates.  Tube-based DMPC
for coupled linear subsystems~\cite{Riverso2012PlugandPlayDM,RIVERSO20142179}
provides local controller synthesis from neighbor information but does
not address high-relative-degree nonlinear safety constraints.  Tube
MPC~\cite{Mayne2005RMPC} provides principled constraint satisfaction
via ancillary feedback and RPI sets for single agents.  Control barrier
functions (CBFs) with exponential/high-order
extensions~\cite{NguyenSreenath2016ECBF,XiaoBelta2019HOCBF} and robust
variants~\cite{XuAmes2015RobustCBF,Jankovic2018RobustCBF} provide
forward-invariance under disturbances.  MPC--CBF
integrations~\cite{Zeng2021MPCDiscreteCBF,Liu2023IterativeMPCDHOCBF} combine predictive
performance with invariance but are predominantly single-agent or
centralized.  A recent distributed safety-critical MPC for nonlinear
MAS formation~\cite{wang2025distributedsafetycriticalmpcmultiagent} does not address model-based leader
feedforward or tube-based robustness margins.

The authors' prior work~\cite{koulong2025wc} established a distributed safety-critical framework for nonlinear multi-agent systems under bounded disturbances. By employing an ancillary feedback law to regulate the deviation between an agent's true state and its recursively planned nominal trajectory, robust positively invariant ellipsoidal tubes are constructed. Their support functions, evaluated along exponential control barrier function (eCBF) gradients, yield tightening margins that ensure nominal constraint satisfaction mathematically guarantees true-trajectory safety. However, applying this general-purpose framework to cooperative leader-follower tracking exposes structural limitations. Treating complex internal dynamics as lumped disturbances without model-based compensation yields a tube radius heavily constrained by the agent's Lipschitz constant, generating increasingly conservative bounds for highly nonlinear systems. Furthermore, this prior approach lacks theoretical mechanisms to bound the global formation error across the communication~graph.

By restricting the problem scope specifically to leader-follower tracking, the present work exploits exact dynamic models of both the leader and followers to achieve fundamentally tighter bounds. First, a model-based feedforward term precisely cancels the follower's internal dynamics from the error evolution and injects the anticipated leader maneuver. By isolating the leader's trajectory deviation as the sole remaining nonlinear disturbance, this architecture completely circumvents the restrictive, Lipschitz-dependent Lyapunov feasibility conditions of~\cite{koulong2025wc}. The safety tube radius is thus transformed from an implicit inequality into an explicit, closed-form expression that directly shrinks required safety margins and reduces conservatism. Furthermore, by formulating the tracking objective via synchronization errors with prescribed inter-agent offsets, we leverage the network topology to derive a closed-form global formation error bound, rigorously connecting individual local tube radii to collective multi-agent performance via the minimum singular value of the augmented communication graph.

The remainder of the paper is structured as follows:
Section~\ref{sec:Preliminaries} states the
system and graph. Section~\ref{sec:Methodology} develops the feedforward
compensation and its impact on the tube radius. Section~\ref{sec:MainResult} develops the formation tracking objective
and global error bound. It also states the tightened eCBF constraints. Section~\ref{sec:NUMERICALEXAMPLE} presents the MPC implementation and demonstrates via simulation. Section~\ref{sec:CONCLUSION} summarizes contributions and potential future work.

\section{Problem Formulation}\label{sec:Preliminaries}

We consider $N$ follower agents $i \in \mathcal{V} = \{1,\ldots,N\}$ and a leader $0$. Each follower has $n$-th order nonlinear Brunovsky
dynamics:
\begin{align}\label{eq:ct_follower}
\dot{x}_p^{i} &= x_{p+1}^{i}, &p \in \{1,\dots,n-1\}, \notag\\
\dot{x}_n^{i} &= f^{i}(x^{i},t) + u^{i} + w^{i}, &w^{i}(t)\in\mathcal D_i,
\end{align} 
where $x_p^i\in\mathbb{R}^d$, $u^i\in\mathbb{R}^d$ is the input, and $w^i\in\mathbb{R}^d$ is a bounded time-varying disturbance with $\|w^{i}(t)\|\le \bar w^i$.
The leader dynamics satisfies:
\begin{align}\label{eq:ct_leader}
\dot{x}_p^{\,0} &= x_{p+1}^{\,0},\hspace{85pt} p \in \{1,\dots,n-1\}, \notag\\
\dot{x}_n^{\,0} &= f^{\,0}(x^{\,0},t) + w^{0},\hspace{40pt} w^{0}(t)\in\mathcal D_0,
\end{align}
where $w^0\in\mathbb{R}^d$ is a bounded time-varying disturbance of leader agent $0$. Both $f^i$ and $f^0$ are locally
Lipschitz with constants $L_i$ and $L_0$ over some acompact domain. 

The agents communicate over a fixed weighted graph $\mathcal{G} = (\mathcal{V},\mathcal{E},A)$ with weighted adjacency $A=[a_{ij}]$, where $a_{ij}>0$ iff $(j,i)\in\mathcal{E}$. The Laplacian is $L:=D-A$. 

The neighbor set of agent $i$ is
$\mathcal{N}_i=\{j: a_{ij}>0\}$. Each agent $i$ has perpetual access to its own full state
$x^i(t)\in\mathbb{R}^{nd}$ 
as well as
the current
state $x^j(t)$ of every 
neighbor $j\in\mathcal{N}_i$ 
over the (fixed) graph
$\mathcal{G}$.

To incorporate a leader agent, we define the augmented graph $\bar{\mathcal{G}}=(\bar{\mathcal{V}},\bar{\mathcal{E}}{\color{black}})$ with $\bar{\mathcal{V}}=\{0,1,\dots,N\}$, which augments the leader interaction matrix $B_0=\mathrm{diag}\{b_{i0}\}$.

\begin{assumption}\label{ass:graph}
The augmented graph $\bar{\mathcal{G}}$ with node set
$\{0\}\cup\mathcal{V}$ contains a spanning tree rooted at the leader.
\end{assumption}

The above assumption is a design-stage condition verified before deployment; it
implies $\nu_1 L+\nu_2 B_0$ is a nonsingular
M-matrix~\cite{koulong2025acc} for any $\nu_1,\nu_2>0$.

The objective of each follower $i$ is to eventually drive the formation synchronization error $(x^i_p - x^0_p - \psi^i_p)$ to zero, robustly tracking the leader's spatial state $x^0_p$ with a prescribed offset $\psi^i_p$ under bounded disturbances and hard safety constraints. Our prior work~\cite{koulong2025wc} established that, under the ancillary control law $u^{\,i} = \bar{u}^i - K_i (x^{\,i} - \bar{x}^{\,i})$ driven by a nominal planner $\bar{u}^i$, an ellipsoidal robust positively invariant (RPI) tube $\mathcal{Z}_i := \{z^i : (z^i)^\top P_i z^i \leq r_i^2\}$ can be constructed with radius:
\begin{align}\label{eq:rho_23}
 r_i \geq \frac{2\bar{w}_{i}\lambda_{\max}(P_{i})}{\lambda_{\min}(Q_{i})-2L_{f^{i}}\lambda_{\max}(P_{i})},
\end{align}
where $P_i \succ 0$ and $Q_i \succ 0$ are associated with the feedback gain $K_i$ through a Lyapunov equation. The support function of $\mathcal{Z}_i$ along each exponential control barrier function (eCBF) gradient yields exact tightening margins ensuring nominal safety implies true-trajectory safety. However, \eqref{eq:rho_23} requires the implicit feasibility constraint $\lambda_{\min}(Q_i) > 2L_{f^i}\lambda_{\max}(P_i)$. Because the Lipschitz constant $L_{f^i}$ appears subtractively in the denominator, highly nonlinear dynamics drive this denominator toward zero, forcing conservative, excessively large tube radii $r_i$ that restrict physical performance.

By focusing specifically on leader-follower tracking with known explicit models, the present work resolves this limitation. We augment the applied control with a model-based feedforward term, yielding $u^{\,i} = \bar{v}^i + f^0(\hat x^0,t) - f^i(x^i,t) - K_i (x^{\,i} - \bar{x}^{\,i})$. This exact dynamic cancellation yields a new, explicitly defined tube radius:
\begin{equation}
  r_i \geq \frac{2 \bar{w}^i_{\mathrm{eff}} \lambda_{\max}(P_i)^{3/2}\|G\|}
  {\lambda_{\min}(Q_i)}.
\end{equation}
While the general-purpose framework of~\cite{koulong2025wc} must conservatively bound the agents' inherent nonlinearities as generic lumped disturbances, the proposed feedforward compensation actively exploits the tracking objective to perfectly isolate the deviation dynamics from these known models. Consequently, the only remaining disturbance sources are the follower's bound $w^i$ and the leader's disturbance $w^0$, which propagates via the mismatch $f^0(x^0,t)-f^0(\hat{\bar x}^0,t)$. Both effects are absorbed into a fixed scalar $\bar w^i_{\mathrm{eff}}$. This completely eliminates the problematic $L_{f^i}$ dependency in~\eqref{eq:rho_23}, yielding an explicit closed-form radius free from restrictive feasibility conditions (Proposition~\ref{prop:tube}). Since implementing this requires access to the leader's state $x^0(t)$, either directly or via network relays, we impose the following assumption.

\begin{assumption}\label{ass:leader_model}
The leader dynamics model $f^0$ is known to all followers at 
design time. Agents with $b_{i0}>0$ have additional direct access to the 
leader state $x^0(t)$,
\[
  \hat{x}^0(t) :=
  \begin{cases}
    x^0(t), & b_{i0}>0 \quad\text{(direct leader access)},\\
    \text{relayed estimate}, & b_{i0}=0.
  \end{cases}
\]
\end{assumption}


For agents with $b_{i0}=0$, Assumption~\ref{ass:graph} guarantees the existence of a finite multi-hop communication path of length~$l$ to the leader; these agents maintain a propagated estimate $\hat{x}^0_i(t)$ forwarded along this path. The resulting feedforward mismatch satisfies $\|f^0(\hat{x}^0_i,t)-f^0(x^0,t)\|\leq L_0 e^0_i$ for a bounded $e^0_i$, which equals zero for directly connected agents and is absorbed into the disturbance bound $\bar{w}^i$ of agent~$i$.

Assumptions~\ref{ass:graph}--\ref{ass:leader_model} are all
design-time verifiable or satisfied by standard onboard sensing.
They govern the continuous-time problem setup only.  The discretization
of the dynamics and the inter-agent plan exchange protocol used for
multi-step prediction inside the OCP are implementation-level details
introduced in Section~\ref{sec:Methodology}.
 
\section{Methodology}\label{sec:Methodology}

\subsection{Feedforward Compensation} To  anticipate the leader’s full dynamics and overcome tracking delays inherent in position-only sensing \cite{Li2012}, each follower applies: 
\begin{equation}\label{eq:feedfwd_input} u^{i}(t) = v^{i}(t) + f^{0}(\hat{x}^0,t) - f^{i}(x^{i},t), 
\end{equation} 
where $-f^i(x^i,t)$ cancels the follower's own nonlinear dynamics,
$f^0(\hat{x}^0,t)$ feedforwards the full leader state estimate available to agent $i$ (equal to $f^{0}(x^0,t)$ exactly when $b_{i0}>0$, and using the relayed estimate $\hat{x}^0$ when $b_{i0}=0$), and $v^i$ is the feedback input designed below.

\subsection{Impact on the Deviation Dynamics}\label{sub:stacked_tube_err}
The ancillary feedback from \cite{koulong2025wc} is:
\begin{align}\label{eq:ancillary_control}
\hspace{-0.5em} v^{i} &= \bar v^{i} -\sum_{p=1}^n K_p^{i}\,\delta x^{i}
      \;=\; \bar v^{i}-K^{i}\delta x^{i}, \quad \delta x^i := x^i-\bar x^i
\end{align}
where $\bar{v}^i$ is the nominal MPC input, $\bar x^i$ is the nominal state of $i$, $K_p^{i}\in\mathbb{R}^{d\times d}$ is the ancillary gain and
$K^{i}\in\mathbb{R}^{d\times nd}$ is the stacked representation.
The nominal
trajectory $\bar{x}^i$ satisfies:
\begin{equation}\label{eq:nominal_dyn}
  \dot{\bar{x}}^i_p = \bar{x}^i_{p+1},\quad
  \dot{\bar{x}}^i_n = \bar{v}^i + f^0(\hat{\bar{x}}^0,t),
\end{equation}
where the feedforward $f^0(\hat{\bar{x}}^0,t)$ uses the nominal leader trajectory 
$\hat{\bar{x}}^0$ propagated from $\hat{x}^0$.
The full true closed-loop dynamics after substituting the feedforward control law \eqref{eq:feedfwd_input} and the ancillary feedback \eqref{eq:ancillary_control} is
\begin{align}\label{eq:ct_follower0}
\dot{x}_p^{i} &= x_{p+1}^{i}, \quad p \in \{1,\dots,n-1\}, \notag\\
\dot{x}_n^{i} &= \bar v^{i}-K^{i}\delta x^{i} + f^{0}(\hat{x}^{0},t)  + w^{i},
\end{align} 
with $\delta x^i_p = x^i_p - \bar x^i_p$, we get
\begin{align}\label{eq:ct_follower1}
\delta\dot{x}^i_p &= \delta x^i_{p+1}, \quad p \in \{1,\dots,n-1\}, \notag\\
\delta\dot{x}^i_n &= -K^i_p\delta x^i + \bigl[f^0(\hat{x}^0,t) - f^0(\hat{\bar{x}}^0,t)\bigr] + w^i.
\end{align} 
In compact form
\begin{equation}
  \delta\dot{x}^i = (A_0 - GK^i_p)\,\delta x^i
  + G\bigl[\underbrace{f^0(\hat{x}^0,t)-f^0(\hat{\bar{x}}^0,t)}_{\text{leader deviation}}
  + w^i\bigr],
  \label{eq:devdyn}
\end{equation}
where $A_0^i \in \mathbb{R}^{nd \times nd}$ is the chain-of-integrators block matrix and $G^i \in \mathbb{R}^{nd \times d}$ selects
the highest-order block. The follower nonlinearity $f^i$ does
not appear in~\eqref{eq:devdyn}: it is cancelled exactly by the
feedforward.
\[
A_0^{i} =
\begin{bmatrix}
0 & I_d & 0 & \cdots & 0 \\
0 & 0 & I_d & \cdots & 0 \\
\vdots & \vdots & \ddots & \ddots & \vdots \\
0 & 0 & \cdots & 0 & I_d \\
0 & 0 & \cdots & 0 & 0
\end{bmatrix},\qquad
G^{i} =
\begin{bmatrix}
0\\ \vdots\\ 0\\ I_d
\end{bmatrix}.
\]

\subsection{Leader Tube Radius}
The leader deviation $\delta x^0:=x^0-\bar{x}^0$ is bounded by the
leader tube radius $\bar{r}_0$.  This follows from applying the RPI
analysis of~\cite{koulong2025wc} to the leader system with disturbance $w^0$ ($\bar{w}^0$ is a design-time parameter chosen to cover the worst case satisfying $\|w^0(t)\| \leq \bar{w}^0$ for all $t \geq 0$) and Lipschitz constant $L_0$:
\begin{equation}
  \bar{r}_0 = \frac{{\rho_0}}{\sqrt{\lambda_{\min}(P_0)}},\quad
  \rho_0\geq\frac{2\lambda_{\max}(P_0)\|G\|\bar{w}^0}
  {\lambda_{\min}(Q_0)-2\lambda_{\max}(P_0)\|G\|L_0},
  \label{eq:r0}
\end{equation}
under $\lambda_{\min}(Q_0)>2\lambda_{\max}(P_0)\|G\|L_0$, giving
$\|\delta x^0_p(t)\|\leq\bar{r}_0$ for all $t\geq t_0$.

\subsection{Per-Agent Tube Radius under Feedforward}

\begin{proposition}[Explicit Tube Radius]\label{prop:tube}
Let $K^i_p$ make $A^i_p:=A_0-GK^i_p$ Hurwitz with
Lyapunov pair $Q_i,P_i\succ 0$ satisfying
$(A^i_p)^\top P_i+P_iA^i_p=-Q_i$.  We define the effective
disturbance bound:
\begin{equation}
  \bar{w}^i_{\mathrm{eff}} := \bar{w}^i + L_0\bar{r}_0.
  \label{eq:weff}
\end{equation}
Then the ellipsoidal tube
$\mathcal{Z}^i:=\{\delta x^i:(\delta x^i)^\top P_i\delta x^i
\leq r_i^2\}$ is robustly positively invariant
under~\eqref{eq:ancillary_control}--\eqref{eq:devdyn} for:
\begin{equation}
  r_i \geq \frac{2\lambda_{\max}(P_i)^{3/2}\|G\|\,\bar{w}^i_{\mathrm{eff}}}
  {\lambda_{\min}(Q_i)}.
  \label{eq:ri}
\end{equation}
The true state deviation satisfies $\|\delta x^i(t)\| \leq
r_i/\sqrt{\lambda_{\min}(P_i)} =: r^i_{\mathrm{ball}}$
for all $t \geq t_0$.
The radius $r_i$ is computed entirely from agent $i$'s own
Lyapunov parameters, $\bar{w}^i$, and the shared scalar $\bar{r}_0$;
no Lipschitz constant $L_i$ appears, in contrast
to~\cite{koulong2025wc}.
\hfill$\square$
\end{proposition}

\begin{proof}
From \eqref{eq:devdyn}, the feedforward cancels $f^i$ exactly, leaving: 
\[ 
\delta\dot{x}^i = A^i_p \delta x^i + G d^i(t), \qquad A^i_p := A_0 - G K^i_p 
\]
where the lumped disturbance $d^i := f^0(\hat{x}^0,t) - f^0(\hat{\bar{x}}^0,t) + w^i$ satisfies: 
\[
\|d^i(t)\| \leq L_0\bar{r}_0 + \bar{w}^i =: \bar{w}^i_{\mathrm{eff}} 
\]
Since $K^i_p$ renders $A^i_p$ Hurwitz, let $P_i, Q_i \succ 0$ solve $(A^i_p)^\top P_i + P_i A^i_p = -Q_i$. With $V_i = (\delta x^i)^\top P_i \delta x^i$, substituting the deviation dynamics: 
\[ 
\dot{V}_i = -(\delta x^i)^\top Q_i\,\delta x^i + 2(\delta x^i)^\top P_i G d^i(t) 
\]
\[ 
\dot{V}_i \leq -\lambda_{\min}(Q_i)\|\delta x^i\|^2 + 2\lambda_{\max}(P_i)\|G\|\bar{w}^i_{\mathrm{eff}}\|\delta x^i\|. 
\]
Since $f^i$ does not appear, the leading coefficient
$-\lambda_{\min}(Q_i) < 0$ unconditionally — no Lipschitz
feasibility condition is required.

The set $\mathcal{Z}^i$ is RPI if $\dot{V}_i \leq 0$ on
$\partial\mathcal{Z}^i$, which requires:
\[
\lambda_{\min}(Q_i)\|\delta x^i\| \geq
2\lambda_{\max}(P_i)\|G\|\bar{w}^i_{\mathrm{eff}}.
\]
On $\partial\mathcal{Z}^i$, $V_i = r_i^2$ and
$V_i \leq \lambda_{\max}(P_i)\|\delta x^i\|^2$, so:
\[
\|\delta x^i\| \geq \frac{r_i}{\sqrt{\lambda_{\max}(P_i)}}.
\]
Substituting this lower bound as a sufficient condition and
solving for $r_i$:
\[
r_i \geq \frac{2\lambda_{\max}(P_i)^{3/2}\|G\|\bar{w}^i_{\mathrm{eff}}}
             {\lambda_{\min}(Q_i)}.
\]
With $r_i$ satisfying~\eqref{eq:ri}, $\dot{V}_i \leq 0$ on
$\partial\mathcal{Z}^i$, so $\mathcal{Z}^i$ is RPI.
Finally, from $V_i \leq r_i^2$ and
$V_i \geq \lambda_{\min}(P_i)\|\delta x^i\|^2$:
\[
\|\delta x^i(t)\| \leq \frac{r_i}{\sqrt{\lambda_{\min}(P_i)}}
=: r^i_{\mathrm{ball}}, \quad \forall\, t \geq t_0.
\]
\end{proof}

\begin{remark}[Structural difference from~\cite{koulong2025wc}]
Without feedforward~\eqref{eq:feedfwd_input}, the deviation dynamics
of~\cite{koulong2025wc} contain $f^i(x^i,t)-f^i(\bar{x}^i,t)$,
bounded by $L_{f^i}\|\delta x^i\|$, which contributes
$2L_{f^i}\lambda_{\max}(P_i)\|\delta x^i\|^2$ to $\dot{V}_i$ and
reduces the coefficient of $\|\delta x^i\|^2$ in the Lyapunov
derivative bound to $\lambda_{\min}(Q_i) - 2L_{f^i}\lambda_{\max}(P_i)$.
For this to be positive --- so that $\dot{V}_i$ can be driven
negative on $\partial\mathcal{Z}^i$ --- the feasibility condition
$\lambda_{\min}(Q_i) > 2L_{f^i}\lambda_{\max}(P_i)$ must hold, and
the resulting radius $\rho_i$ grows unboundedly as $L_{f^i}$
approaches this limit.
With the feedforward~\eqref{eq:feedfwd_input}, the $f^i$ term is
absent from~\eqref{eq:devdyn}; the coefficient of $\|\delta x^i\|^2$
in $\dot{V}_i$ reduces to $\lambda_{\min}(Q_i) > 0$ unconditionally,
so any $K^i_p$ rendering $A^i_p$ Hurwitz suffices, and the tube
radius becomes the explicit closed-form expression~\eqref{eq:ri}
with no condition on $L_{f^i}$.
\end{remark}

\begin{remark}[Quantitative improvement over~\cite{koulong2025wc}]
In~\cite{koulong2025wc}, $f^i$ appears in the deviation dynamics,
contributing $2L_{f^i}\lambda_{\max}(P_i)\|z^i\|^2$ to $\dot{V}_i$.
This reduces the effective coefficient to
$\lambda_{\min}(Q_i) - 2L_{f^i}\lambda_{\max}(P_i)$, which must
remain positive, imposing the feasibility condition
$\lambda_{\min}(Q_i) > 2L_{f^i}\lambda_{\max}(P_i)$, and yields:
\[
     r^{\mathrm{no\,FF}}_i \geq \frac{2\bar{w}_i\lambda_{\max}(P_i)}
    {\lambda_{\min}(Q_i) - 2L_{f^i}\lambda_{\max}(P_i)}.
\]
This radius grows unboundedly as $L_{f^i}$ approaches the
feasibility limit and must be verified to be finite for each agent.
Under the feedforward~\eqref{eq:feedfwd_input}, $f^i$ is absent
from the deviation dynamics, the coefficient reduces to
$\lambda_{\min}(Q_i) > 0$ unconditionally, and~\eqref{eq:ri} gives:
\[
    r^{\mathrm{FF}}_i \geq \frac{2\lambda_{\max}(P_i)^{3/2}\|G\|
    \bar{w}^i_{\mathrm{eff}}}{\lambda_{\min}(Q_i)},
\]
with no condition on $L_{f^i}$ and no risk of the denominator
vanishing, regardless of how nonlinear the follower dynamics are. 
\end{remark}

\begin{remark}[Conservatism of worst-case tube design]
The tube radius $r^i$ is sized for the worst-case disturbance bound 
$\bar{w}^i$ uniformly over all time, which is conservative when the 
actual $w^i(t)$ is significantly smaller than $\bar{w}^i$ during 
typical operation.  Adaptive or time-varying tube designs that exploit 
online disturbance measurements could reduce this conservatism but 
would require the scalar $r^j(t_k)$ to be re-broadcast at every 
timestep, increasing communication overhead.  The fixed worst-case 
design is retained here for simplicity and to preserve the one-time 
broadcast property.
\end{remark}

\subsection{Formation Tracking and Global Error Bound}

\subsubsection{Synchronization Error and Formation Objective}

The local synchronization error for agent $i$ at state order $p$ is:
\begin{multline}\label{eq:local_syn_err}
e_p^i := - \nu_1 \sum_{j\in\mathcal N_i} a_{ij}
\big[(x_p^i - \psi_p^i) - (x_p^j - \psi_p^j)\big]  \\
- \nu_2 b_{i0}
\big[(x_p^i - \psi_p^i) - (x_p^0 - \psi_p^0)\big],
\end{multline}
where $\psi^i_p\in\mathbb{R}^d$ is the prescribed formation offset for
agent $i$ at order $p$. For agents with $b_{i0}>0$, the second term in
\eqref{eq:local_syn_err} enforces direct synchronization to the leader.  For agents with $b_{i0}=0$, this term vanishes, so $e^i_p$ is computed
solely from the agent's own state and its neighbors' states; no direct
knowledge of $x^0$ is required.  Thus, \eqref{eq:local_syn_err} is a graph-based synchronization error, not a per-agent estimate error to the leader.  The locally available leader signal $\hat{x}^0$ enters
instead through the feedforward law~\eqref{eq:feedfwd_input}, the nominal
leader propagation~\eqref{eq:nominal_dyn}, and the local deviation
dynamics~\eqref{eq:devdyn}. When $e^i_p=0$ for all $i$ and $p$, all
agents maintain their desired relative positions to neighbors and the
leader.  The nominal counterpart $\bar{e}^i_p$ replaces true states
with nominal trajectories; under~\eqref{eq:feedfwd_input} the nominal error dynamics reduce to a chain-of-integrators form driven by the optimizer input $\bar{v}^i$ and neighbor plan predictions.  The MPC
cost (Section~\ref{sec:MainResult}) penalizes $\bar{e}^i$ to drive the
formation error to zero.

\subsubsection{Connecting Local Tube Radii to Global Formation
Error}

The true synchronization error at order $p$ is given by ~\eqref{eq:local_syn_err} and its corresponding nominal synchronization error $\bar{e}^i_p$ has the same structure but with nominal states. Computing $e^i_p - \bar{e}^i_p$ term by term, the formation offsets $\psi^i_p$, $\psi^j_p$, $\psi^0_p$ are constant and cancel exactly:
\begin{multline}
 e^i_p - \bar{e}^i_p = -\nu_1\sum_{j\in\mathcal{N}_i}a_{ij}\bigl[(x^i_p - \bar{x}^i_p) - (x^j_p - \bar{x}^j_p)\bigr] \\ - \nu_2 b_{i0}\bigl[(x^i_p - \bar{x}^i_p) - (x^0_p - \bar{x}^0_p)\bigr].   
\end{multline}
Substituting the deviation definitions $\delta x^i_p := x^i_p - \bar{x}^i_p$, $\delta x^j_p := x^j_p - \bar{x}^j_p$, $\delta x^0_p := x^0_p - \bar{x}^0_p$:
\begin{equation}
e^i_p - \bar{e}^i_p = -\nu_1\sum_{j\in\mathcal{N}_i}a_{ij}\bigl(\delta x^i_p - \delta x^j_p\bigr) - \nu_2 b_{i0}\bigl(\delta x^i_p - \delta x^0_p\bigr)
\end{equation}
Expanding and collecting terms in $\delta x_p^i$ gives:
\begin{multline}
e^i_p - \bar{e}^i_p = -(\nu_1 \sum_{j\in\mathcal{N}_i}a_{ij} + \nu_2 b_{i0}) \delta x^i_p \\ + \nu_1\sum_{j\in\mathcal{N}_i}a_{ij}\delta x^j_p + \nu_2 b_{i0}\delta x^0_p,
\end{multline}
which is the exact linear map from physical state deviations to synchronization error deviation.

\begin{lemma}[Synchronization-Error Deviation Bound]\label{lem:zbound}
For all $t\geq t_0$ and all $p\in\{1,\ldots,n\}$:
\begin{equation}
  \|e^i_p(t) - \bar{e}^i_p(t)\| \leq \bar{z}^i.
  \label{eq:zboundresult}
\end{equation}
\end{lemma}
\begin{proof}
Since $e^i_p - \bar{e}^i_p = -(\nu_1 \sum_{j\in\mathcal{N}_i}a_{ij} +\nu_2 b_{i0})\delta x^i_p
+\nu_1\sum_{j\in\mathcal{N}_i} a_{ij}\delta x^j_p+\nu_2 b_{i0}\delta x^0_p$, applying
the triangle inequality and $\|\delta x^i_p\|\leq r_i$,
$\|\delta x^j_p\|\leq r_j$, $\|\delta x^0_p\|\leq\bar{r}_0$, then
$
\|e^i_p - \bar{e}^i_p\| \leq (\nu_1 d_i + \nu_2 b_{i0})\|\delta x^i_p\| + \nu_1\sum_{j\in\mathcal{N}_i}a_{ij}\|\delta x^j_p\| + \nu_2 b_{i0}\|\delta x^0_p\|
$
and
$
\|e^i_p - \bar{e}^i_p\| \leq (\nu_1 d_i + \nu_2 b_{i0})r_i + \nu_1\sum_{j\in\mathcal{N}_i}a_{ij}r_j + \nu_2 b_{i0}\bar{r}_0 =: \bar{z}^i
$ which
gives~\eqref{eq:zboundresult}.
\end{proof}

The bound $\bar{z}^i$ requires only agent $i$'s own radius $r_i$,
the leader radius $\bar{r}_0$, and the neighbor scalars
$\{r_j\}_{j\in\mathcal{N}_i}$.  Each $r_j$ is computed locally by
agent $j$ using~\eqref{eq:ri} and broadcast to direct neighbors at initialization.

\begin{theorem}[Global Formation Error Bound]\label{thm:formation}
Let ${sv}_* \;:=\; \inf_{\sigma\in\Sigma}\; \sigma_{\min}\!\big(\nu_1 L+\nu_2 B_0\big)\;>\;0$ (guaranteed by
Assumption~\ref{ass:graph}).  If the nominal errors satisfy
$\|\bar{e}^i(t)\|_2\to 0$ as $t\to\infty$ for all $i$ (ensured by
the terminal ingredients of Assumption~\ref{ass:terminal}), then:
\begin{multline}\label{eq:globalbound}
\limsup_{t \to \infty} \left(\sum_{i=1}^{N} \|(x_p^i-\psi_p^i) - (x_p^0-\psi_p^0)\|_2^2\right)^{1/2}  
\\ \le \frac{1}{{sv}_*} \left( \sum_{i=1}^N (\bar{z}^i)^2 \right)^{1/2}.
\end{multline}
\hfill $\square$
\end{theorem}

\begin{proof}
Stacking~\eqref{eq:local_syn_err} yields $e_p = -(\nu_1 L+\nu_2 B_0)q_p$,
where $q_p$ collects $(x^i_p-\psi^i_p)-(x^0_p-\psi^0_p)$ and $e_p$
collects $e^i_p$.  Since $\nu_1 L+\nu_2 B_0$ is nonsingular under
Assumption~\ref{ass:graph},
$\|q_p\|_2\leq\sigma_{\min}^{-1}(\nu_1 L+\nu_2 B_0)\|e_p\|_2$.
Using $\|e^i_p\|_2\leq\|\bar{e}^i_p\|_2+\bar{z}^i$ from
Lemma~\ref{lem:zbound}, $\|\bar{e}^i_p\|_2\to 0$, and taking
$\limsup$ gives~\eqref{eq:globalbound}.
\end{proof}

\begin{remark}
The
bound~\eqref{eq:globalbound} is entirely determined by $\{r_i\}$ and
the fixed graph structure; tighter tube radii from the feedforward
directly reduce the formation error bound.
\end{remark}

\subsection{Tightened eCBF Functions}
Let nominal dynamics be control-affine $\dot{\bar{x}}^i=F_x(\bar{x}^i,t)+F_u\bar{v}^i$. The support-function tightening mechanism of~\cite{koulong2025wc}
is applied to both inter-agent avoidance and obstacle collision
avoidance, respectively:
\begin{multline}\label{eq:eCBF_col_final}
\hspace{-9pt} \Phi_{ij}^{\text{tight}}\big(\bar{x}^{i}, \bar{x}^{j}, \bar{v}^{i};\sigma\big) =L_{F_{x}}^{r} h_{ij}(\bar{x}^{i},\bar{x}^{j}) + \sum_{q=1}^{r-1} \kappa_q L_{F_{x}}^{q} h_{ij}(\bar{x}^{i},\bar{x}^{j}) \\  \hspace{9pt} + \big(L_{F_{u}} L_{F_{x}}^{r-1} h_{ij}(\bar{x}^{i},\bar{x}^{j})\big) \bar{v}^{i}
+\kappa_0 \big( h_{ij} (\bar x^i,\bar x^j) -\delta_{ij}(\bar{x}^{i},\bar{x}^{j})\big),
\end{multline}
\begin{multline}\label{eq:eCBF_obs_final}
\Phi_{iO}^{\text{tight}}\big(\bar{x}^{i}, \bar{v}^{i};\sigma\big) = L_{F_{x}}^r h_{iO}(\bar{x}^{i}) + \sum_{q=1}^{r-1} \kappa_q L_{F_{x}}^{q} h_{iO}(\bar{x}^{i}) 
\\
+ \big(L_{F_{u}} L_{F_{x}}^{r-1} h_{iO}(\bar{x}^{i})\big) \bar{v}^{i} + \kappa_0 \big( h_{iO} (\bar x^i) - \delta_{iO}(\bar{x}^{i})\big),
\end{multline}
where \\ $\delta_{ij}(\bar{x}^{i},\bar{x}^{j}):= r_i\|\nabla_{\bar{x}^i}h_{ij}(\bar{x}^{i},\bar{x}^{j})\|_2+r_j\|\nabla_{\bar{x}^j}h_{ij}(\bar{x}^{i},\bar{x}^{j})\|_2$ and $\delta_{iO}(\bar{x}^{i}):=  r_i\|\nabla_{\bar{x}^i}h_{iO}(\bar{x}^{i})\|_2$.
The term $r_j\|\nabla_{\bar{x}^j}h_{ij}\|_2$ accounts for the
uncertainty in agent $j$'s true state relative to its nominal
prediction $\hat{x}^j$.  It uses the scalar $r_j$ received from
neighbor $j$ at initialization; no runtime communication is needed.

\begin{remark}
The obstacle margin $\delta_{iO}$ and the safety certificate
(nominal satisfaction $\Rightarrow$ true trajectory safety) are
established in~\cite[Thm.~3.1]{koulong2025wc}. The safety guarantee for
\eqref{eq:eCBF_col_final}--\eqref{eq:eCBF_obs_final} follows by applying
\cite[Thm.~3.1]{koulong2025wc} with $\|\delta x^i_p\|\leq r_i$
and $\|\delta x^j_p\|\leq r_j$ simultaneously.
\end{remark}

\begin{theorem}[Robust Formation Tracking and Safety]\label{thm:main}
Consider the system~\eqref{eq:ct_follower}--\eqref{eq:ct_leader} under
Assumptions~\ref{ass:graph}--\ref{ass:leader_model} and the control
law~\eqref{eq:feedfwd_input}--\eqref{eq:ancillary_control}.  Suppose:
\begin{enumerate}
  \item $K^i_p$ makes $A^i_p$ Hurwitz and $r_i$
    satisfies~\eqref{eq:ri};
  \item Nominal errors satisfy $\|\bar{e}^i(t)\|_2\to 0$
    (Assumption~\ref{ass:terminal});
  \item The nominal planner enforces~\eqref{eq:eCBF_col_final}
    and~\eqref{eq:eCBF_obs_final} with margins $\delta_{iO}$,
    $\delta_{ij}$;
  \item Initial conditions satisfy $\delta x^i(t_0)\in\mathcal{Z}^i$,
    $h_{ij}(x^i(t_0),x^j(t_0))\geq 0$, $h_{iO}(x^i(t_0))\geq 0$.
\end{enumerate}
Then for all $t\geq t_0$:
\begin{enumerate}
  \item $\|\delta x^i(t)\|\leq r_i$;
  \item the global formation error bound~\eqref{eq:globalbound} holds;
  \item $h_{ij}(x^i(t),x^j(t))\geq 0$ and $h_{iO}(x^i(t))\geq 0$.
\end{enumerate}
\end{theorem}

\begin{proof}
Claim~1 follows from Proposition~\ref{prop:tube}.  Claim~2 follows
from Theorem~\ref{thm:formation}.  Claim~3: since
$\|\delta x^i_p\|\leq r_i$ by claim~1, the nominal planner satisfying
\eqref{eq:eCBF_col_final}--\eqref{eq:eCBF_obs_final} implies true trajectory safety
by~\cite[Thm.~3.1]{koulong2025wc} applied to each constraint
with margins $\delta_{iO}$ and $\delta_{ij}$.
\end{proof}

\section{MPC Implementation}\label{sec:MainResult}

\subsection{Discretization and Information Exchange}
The continuous-time dynamics~\eqref{eq:ct_follower}--\eqref{eq:ct_leader}
are integrated over each sampling period $T_s$ using a fourth-order
Runge--Kutta scheme.  The resulting distributed MPC architecture
follows the non-iterative plan exchange mechanism of~\cite{DunbarMurray2006}: at each sampling instant
$t_k=kT_s$, all agents solve their local OCPs in parallel using the
previous sampling instant's broadcast from neighbors, then each agent
broadcasts its newly computed nominal input plan to direct neighbors,
and finally all agents apply their control inputs.  This eliminates
any need for iterative agent-to-agent negotiation and ensures no
circular dependency between agents at the same timestep.

\begin{assumption}[Plan exchange]\label{ass:plan}
Following~\cite{DunbarMurray2006}, after solving its local OCP
at sampling instant $t_k$, each agent $j$ broadcasts its nominal
input plan $\Pi^j(t_k):=\{\bar{v}^{j*}_{l|t_k}\}_{l=0}^{H-1}$ to all
direct neighbors $i\in\{m:j\in\mathcal{N}_m\}$ over the fixed graph
$\mathcal{G}$.  Agent $i$ uses the most recently received plan
$\Pi^j(t_{k-1})$ --- the plan from the previous sampling instant ---
when solving its own OCP at $t_k$.  
\end{assumption}

The one-step lag between $\Pi^j(t_{k-1})$ and the current optimal
$\Pi^j(t_k)$ introduces a plan-reconstruction mismatch (analytical quantity used in the theoretical analysis)
$\tilde{v}^j_k := \bar{v}^j_{l|t_k} - \bar{v}^j_{l|t_{k-1}}$, which
is bounded by how much the optimal plan changes between consecutive
sampling instants.  As shown in~\cite{DunbarMurray2006}, this
mismatch is bounded under Lipschitz continuity of the value
function and is accounted for in the tube radius $r_i$
via~\eqref{eq:weff}.

\subsection{Neighbor Prediction Rollout}\label{subsec:neighborpred}

Agent $i$ uses the received plan $\Pi^j(t_{k-1})$ to reconstruct a
time-shifted nominal input sequence $\hat{v}^j_{l|t_k}$ by dropping
the first element of $\Pi^j(t_{k-1}) = \{\bar{v}^{j*}_{l|t_{k-1}}\}_{l=0}^{H-1}$ (it was applied at $t_{k-1}$) and appending the terminal law
$-\hat{K}^j\bar{e}^j_{H|t_{k-1}}$:
\[
\hat{v}^j_{l|t_k} := \begin{cases} \bar{v}^{j*}_{l+1|t_{k-1}}, & l = 0,\ldots,H-2 \\ -\hat{K}^j\bar{e}^j_{H|t_{k-1}}, & l = H-1 
\end{cases}
\]
which is fed into agent $i$'s OCP as an exogenous fixed signal representing what neighbor $j$ is expected to do. It drives the formation error prediction $\bar{e}^i_{l+1|t_k}$ and the inter-agent eCBF constraint $\Phi^{\mathrm{tight}}_{ij}$.

The corresponding nominal state
prediction $\hat{x}^j_{l|t_k}$ is obtained by rolling out
$\hat{v}^j_{l|t_k}$ forward from the current measured state $x^j(t_k)$
using the nominal dynamics of agent $j$.  This rollout requires agent
$i$ to know agent $j$'s dynamics model $f^j$, which is assumed known
at design time for all direct neighbors, consistent with the formation
setting where the vehicle types are specified before deployment.

\subsection{Applied Control and OCP}

The applied control law is:
\begin{equation}
  u^i(t)=\bar{v}^i(t_k)-K^i_p\delta x^i(t)+f^0(\hat x^0,t)-f^i(x^i,t),
  \label{eq:control}
\end{equation}
where $\bar{v}^i\in\mathcal{V}^i:=\mathcal{U}^i\ominus
K^i_p\mathcal{Z}^i$ is the nominal MPC input held constant over
$[t_k,t_{k+1})$, while $-K^i_p\delta x^i(t)$ and the feedforward
are applied continuously.  The optimizer propagates $\bar{e}^i$
(formation error space) for the cost and $\bar{x}^i$ (state space)
for eCBF evaluation.

\begin{assumption}[Terminal Ingredients]\label{ass:terminal}
There exist $P_r\succ 0$, terminal gain $\hat{K}^i$, and terminal set
$\mathcal{X}_{f,i}$ such that $\bar{v}^i=-\hat{K}^i\bar{e}^i$ renders
$\mathcal{X}_{f,i}$ positively invariant within $\mathcal{V}^i$,
satisfies~\eqref{eq:eCBF_col_final}--\eqref{eq:eCBF_obs_final}, and drives
$\|\bar{e}^i\|_2\to 0$.  Such terminal ingredients exist for the
linearized error dynamics and are computed offline.
\end{assumption}

The local OCP for agent $i$ at sampling instant $t_k$ is:
\begin{align}
\hspace{-9em} \min_{\{\bar{v}^{i}_{l|t_k}\}} \, &
\sum_{l=0}^{H-1}\Big[\|\bar{e}^{i}_{l|t_k}\|_{Q_r}^2 + \|\bar{v}^{i}_{l|t_k}\|_R^2
+ \|\Delta \bar{v}^{i}_{l|t_k}\|_{R_\Delta}^2\Big] + \|\bar{e}^{i}_{H|t_k}\|_{P_r}^2
\nonumber\\
\text{s.t.}\quad &
\bar{e}^{i}_{0|t_k} = e^{i}(t), \qquad \bar{x}^{i}_{0|t_k} = \bar{x}^{i}(t), \nonumber\\
& \bar{e}^{i}_{l+1|t_k} = \Phi_{\mathrm{RK4}}^{(e)}\!\Big(\bar{e}^{i}_{l|t_k},\,\bar{v}^{i}_{l|t_k},\,\{\hat v^j_{l|t_k}\}_{j\in\mathcal N_i^{\sigma(t)}},\,\sigma(t),\,T_s\Big),
\nonumber\\
& \bar{x}^{i}_{l+1|t_k} = \Phi_{\mathrm{RK4}}^{(x)}\!\Big(\bar{x}^{i}_{l|t_k},\,\bar{v}^{i}_{l|t_k}+\hat f^0_{l|t_k}
,\,T_s\Big),
\nonumber\\
& (\bar e^i_{l|t_k},\bar v^i_{l|t_k})\in(\mathcal C_i\ominus\mathcal Z_i)\times\mathbb V,
\nonumber\\
& \Phi_{ij}^{\mathrm{tight}}\!\Big(\bar{x}^{i}_{l|t_k},\,\hat{x}^{j}_{l|t_k},\,\bar{v}^{i}_{l|t_k};\,\sigma(t)\Big)\ge 0,\quad \forall j\in\mathcal N_i^{\sigma(t)},
\nonumber\\
& \Phi_{iO}^{\mathrm{tight}}\!\Big(\bar{x}^{i}_{l|t_k},\,\bar{v}^{i}_{l|t_k};\,\sigma(t)\Big)\ge 0,\quad \forall O\in\mathcal O^{\mathrm{act}},
\nonumber\\
& \bar{e}^{i}_{H|t_k}\in\mathcal X_{f,i},\qquad \bar{v}^{i}_{H-1|t_k}=-\hat{K}^i\bar e^i_{H|t_k}.
\label{eq:ocp}
\end{align}
$\bar{v}^i$ is the variable; $\bar{v}^{i*}$ is its optimal value after the OCP is solved. After the OCP is solved, $\bar{v}^{i*}$ is a specific sequence of numbers. The first element $\bar{v}^{i*}_{0|t_k}$ is applied as the held nominal input over $[t_k, t_{k+1})$. The full sequence $\{\bar{v}^{i*}_{l|t_k}\}_{l=0}^{H-1}$ is broadcast to neighbors as $\Pi^i(t_k)$. The reconstructed neighbor input prediction $\{\hat v^j_{l|t_k}\}$ is agent $i$'s best available prediction of what agent $j$ intends to do over the current horizon, reconstructed from the most recently received broadcast. It is constructed locally by agent $i$ from the received broadcast. Agent $j$ never sends $\hat{v}^j$ directly — agent $i$ constructs it as explained in Section~\ref{subsec:neighborpred}.

The tube $\mathcal{Z}^i$, margins $\delta_{iO}$ and $\delta_{ij}$, and
the scalar $r_j$ from each neighbor are all computed at initialization
and fixed thereafter, adding no per-timestep overhead beyond the OCP
solve and the neighbor plan broadcast of
$\Pi^j(t_k):=\{\bar{v}^j_{l|t_k}\}_{l=0}^{H-1}$.

\subsection{Algorithm}

Algorithm~\ref{alg:dmpc} summarizes the complete procedure for agent
$i$, following the non-iterative structure of~\cite{DunbarMurray2006}.

\begin{algorithm}[t]
\caption{Distributed Tube MPC with Tightened eCBFs for Follower $i$}
\label{alg:dmpc}
\small
\textbf{Data:} $x^i(t_0)$, $\{x^j(t_0)\}_{j\in\mathcal{N}_i}$,
$x^0(t_0)$ if $b_{i0}=0$ or $\hat{x}^0(t_0)$ if $b_{i0}>0$, $T_s$, $H$,
$f^i$, $f^0$, $\{f^j\}_{j\in\mathcal{N}_i}$,
$K^i_p$, $Q_i$, $P_i$, $\bar{w}^i$, $\psi^i$.
\medskip

\textbf{Offline design (once before deployment):}
\begin{algorithmic}[1]
  \State Compute leader tube radius $\bar{r}_0$ via~\eqref{eq:r0}.
  \State Compute effective disturbance bound
         $\bar{w}^i_{\mathrm{eff}} \leftarrow \bar{w}^i + L_0\bar{r}_0$
         via~\eqref{eq:weff}.
  \State Compute per-agent tube radius $r_i$ via~\eqref{eq:ri}
         and tube $\mathcal{Z}^i$.
  \State Compute tightened input set
         $\mathcal{V}^i \leftarrow \mathcal{U}^i \ominus K^i_p\mathcal{Z}^i$.
  \State Compute terminal ingredients
         $\hat{K}^i$, $\mathcal{X}_{f,i}$, $P_r$
         (Assumption~\ref{ass:terminal}).
\end{algorithmic}
\medskip

\textbf{Initialization at $t_0$ (one-time broadcast):}
\begin{algorithmic}[1]
  \State Broadcast $r_i$ to all $j$ such that $i\in\mathcal{N}_j$;
         receive $r_j$ from each $j\in\mathcal{N}_i$.
  \State For each $j\in\mathcal{N}_i$, form inter-agent eCBF margin
         $\delta_{ij} \leftarrow
         r_i\|\nabla_{\bar{x}^i}h_{ij}\|_2
         + r_j\|\nabla_{\bar{x}^j}h_{ij}\|_2$.
  \State Solve OCP~\eqref{eq:ocp} at $t_0$ with
         $\hat{v}^j_{l|t_0}=0$ for all $j\in\mathcal{N}_i$,
         $l=0,\ldots,H-1$ (zero-input initialization), $\hat{x}^j_{l|t_0}=x^j(t_0)$
         for all $j\in\mathcal{N}_i$, $l=0,\ldots,H-1$, and with
         the inter-agent eCBF constraint $\Phi^{\mathrm{tight}}_{ij}\geq 0$
         removed (obstacle avoidance $\Phi^{\mathrm{tight}}_{iO}\geq 0$
         is still enforced).
  \State Broadcast initial plan
         $\Pi^i(t_0)\leftarrow\{\bar{v}^{i*}_{l|t_0}\}_{l=0}^{H-1}$
         to all $j$ s.t. $i\in\mathcal{N}_j$.
\end{algorithmic}
\medskip

\textbf{At each sampling instant $t_k$, $k=1,2,\ldots$:}
\begin{algorithmic}[1]
  \State \textbf{Measure:} obtain $x^i(t_k)$ and
         $\{x^j(t_k)\}_{j\in\mathcal{N}_i}$.
  \State \textbf{Leader state:} obtain $x^0(t_k)$ directly if
         $b_{i0}>0$, or update relayed estimate $\hat{x}^0(t_k)$
         via Assumption~\ref{ass:leader_model}.
  \State \textbf{Predict neighbors:} for each $j\in\mathcal{N}_i$,
         reconstruct $\hat{v}^j_{l|t_k}$ from $\Pi^j(t_{k-1})$ by
         dropping the first element and appending
         $-\hat{K}^j\bar{e}^j_{H|t_{k-1}}$; roll out $\hat{x}^j_{l|t_k}$
         from $x^j(t_k)$ using $f^j$ and $\hat{v}^j_{l|t_k}$.
  \State \textbf{Solve OCP:} solve~\eqref{eq:ocp} using
         $\{\hat{v}^j_{l|t_k}\}_{j\in\mathcal{N}_i}$,
         $x^0(t_k)$ (or $\hat{x}^0_i(t_k)$), and tightened constraints
         $\mathcal{Z}^i$, $\delta_{ij}$, $\delta_{iO}$;
         obtain $\{\bar{v}^{i*}_{l|t_k}\}_{l=0}^{H-1}$.
  \State \textbf{Broadcast:} transmit updated plan
         $\Pi^i(t_k)\leftarrow\{\bar{v}^{i*}_{l|t_k}\}_{l=0}^{H-1}$
         to all $j$ such that $i\in\mathcal{N}_j$.
  \State \textbf{Apply:} over $[t_k, t_{k+1})$, apply:
  \[
    u^i(t) = \bar{v}^{i*}_{0|t_k}
            - K^i_p\bigl(x^i(t)-\bar{x}^i(t)\bigr)
            + f^0(\hat x^0(t),t) - f^i(x^i(t),t).
  \]
  where $\hat{x}^0(t)=x^0(t)$ if $b_{i0}>0$, else (Assumption~\ref{ass:leader_model})
\end{algorithmic}
\end{algorithm}

\section{Numerical Example}\label{sec:NUMERICALEXAMPLE}  

We simulate one leader and five followers with distinct heterogeneous nonlinear dynamics and bounded external disturbances over a fixed communication graph similar to the numerical example with parameters $n=3$, $d=2$, $N=5$, $T_s=0.1$\,s, $H=5$, simulation time $30$\,s studied in \cite{koulong2025wc} with the exception that, here, the leader has additional matched time-varying bounded disturbance. Two circular obstacles are placed at $[1,1]$ and $[-1.5,0.5]$ with radii $0.50$ and $0.65$ (plus $0.15$ inflation). The leader starts at $[3,0,0,0,0,0]^\top$ and followers at $[1.0, -0.5, \text{zeros}(1,4)]$,~$[-0.8, 0.4, \text{zeros}(1,4)]$, \\~$[0.6, 0.6, \text{zeros}(1,4)]$, $[-0.5,  -0.7, \text{zeros}(1,4)]$, \\and~$[0.7, 0.0, \text{zeros}(1,4)]$. Formation offsets are $\psi=[[-9;2], [-6;2], [0;2], [6;2], [9;2]]$.
The objective is to compare the results using the tightened safety constrains from \cite{koulong2025wc} with this current proposed framework.

\begin{figure}[htbp]
    \centering
    \includegraphics[width=2.5in]
    {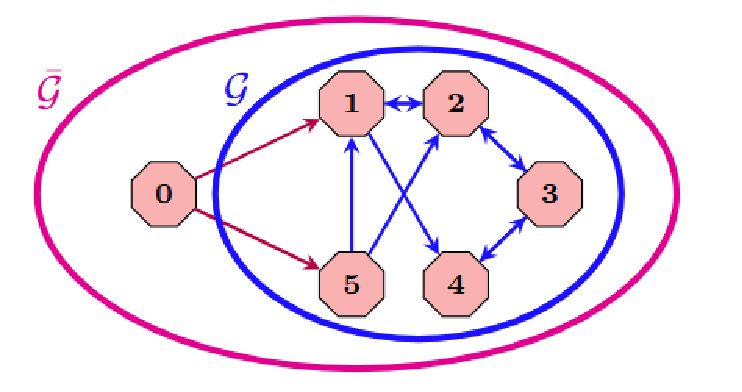}
    \caption{The considered communication topology \({\mathcal{G}}\) and the augmented graph \(\bar{\mathcal{G}}\) for the example in Section~\ref{sec:NUMERICALEXAMPLE}.}
    \label{fig:1}
\end{figure}

\begin{figure}[htbp]
    \centering
        \begin{subfigure}[b]{0.410\textwidth}
    \centering
    \includegraphics[width=2.9in]
    {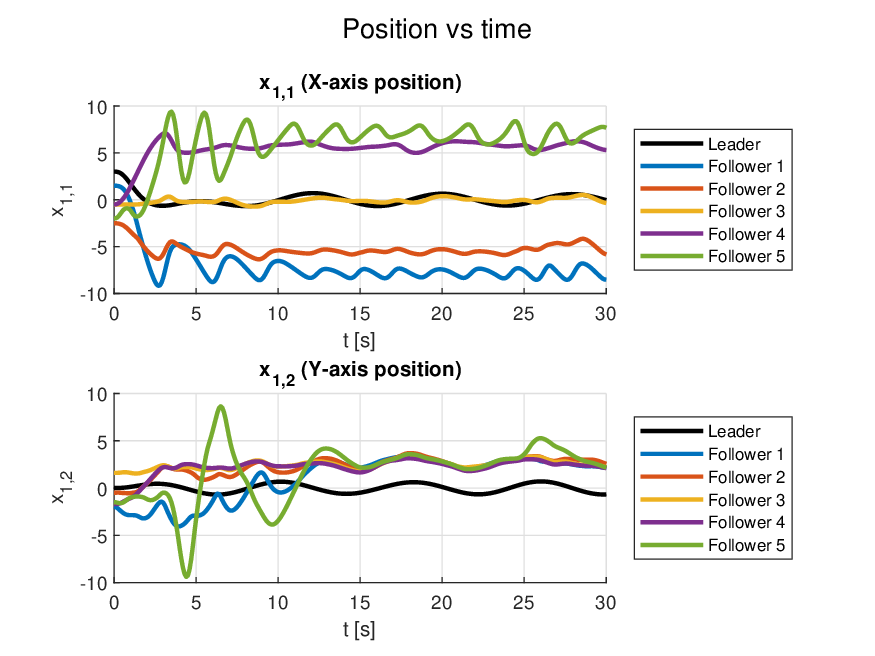}
    \caption{Leader and follower agents $[x^i_{1,1},x^i_{1,2},t]$ position components using the tightened safety constraints of \cite{koulong2025wc}}
    \label{fig:2}
    \end{subfigure}
    \hfill
    \begin{subfigure}[b]{0.410\textwidth}
    \centering
    \includegraphics[width=2.9in]{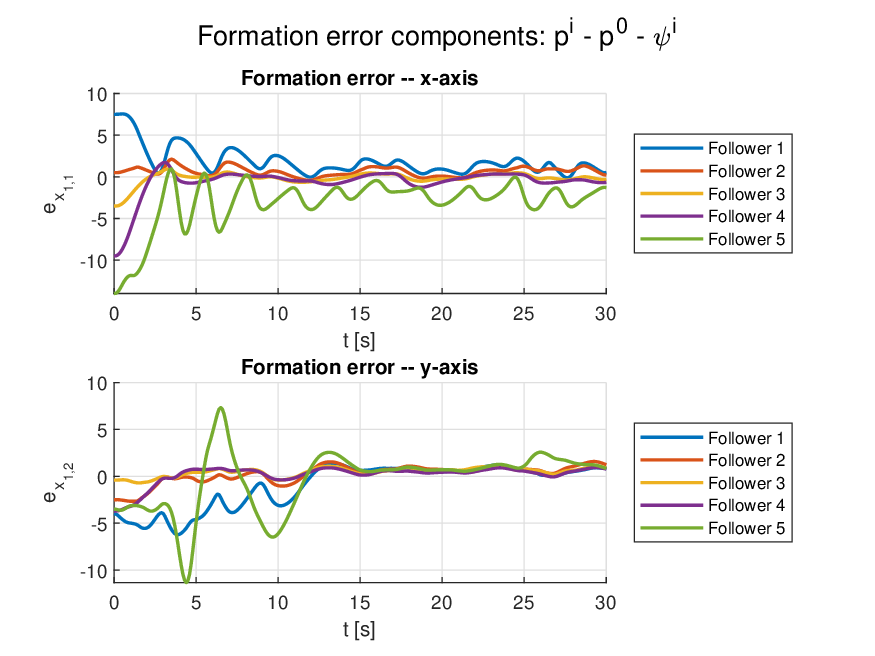}
    \caption{Leader and follower agents $[x_{1,1},x_{1,2}]$ position error components using the tightened safety constraints of \cite{koulong2025wc}}
    \label{fig:4}
    \end{subfigure}
    \hfill
    \caption{Formation offset error
and position components for leader ($i=0$) and followers ($i=1,\ldots,5$)  using the tightened safety constraints of \cite{koulong2025wc}}
    \label{fig:combined1}
\end{figure}

\begin{figure}[htbp]
    \centering
    \begin{subfigure}[b]{0.410\textwidth}
    \centering
    \includegraphics[width=2.9in]{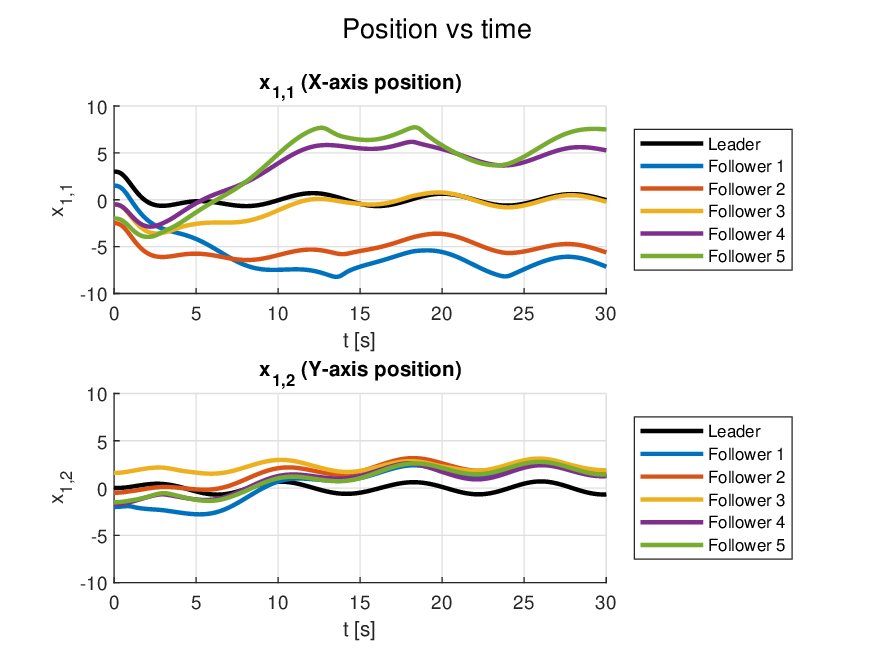}
    \caption{Leader and Follower Agents $[x^i_{1,1},x^i_{1,2},t]$}
    \label{fig:3}
    \end{subfigure}
    \hfill
    \begin{subfigure}[b]{0.410\textwidth}
    \centering
    \includegraphics[width=2.9in]{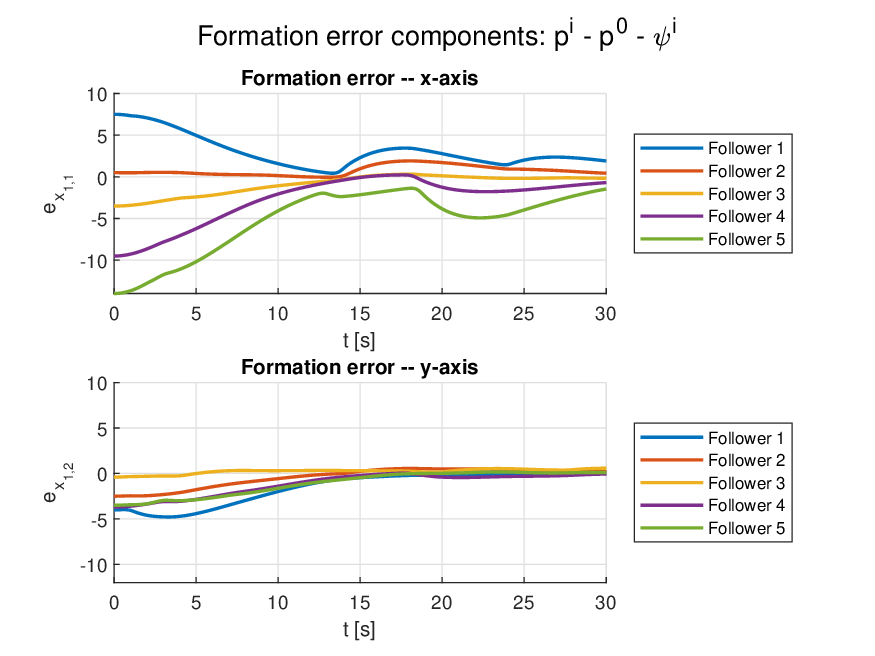}
    \caption{Leader and Follower Agents $[x_{1,1},x_{1,2}]$ Position Error Components}
    \label{fig:5}
    \end{subfigure}
    \hfill
    \caption{Formation offset error
and position components for leader ($i=0$) and followers ($i=1,\ldots,5$)} 
    \label{fig:combined0}
\end{figure}

\begin{figure}[htbp]
    \centering
    \begin{subfigure}[b]{0.410\textwidth}
    \centering
    \includegraphics[width=2.9in]
    {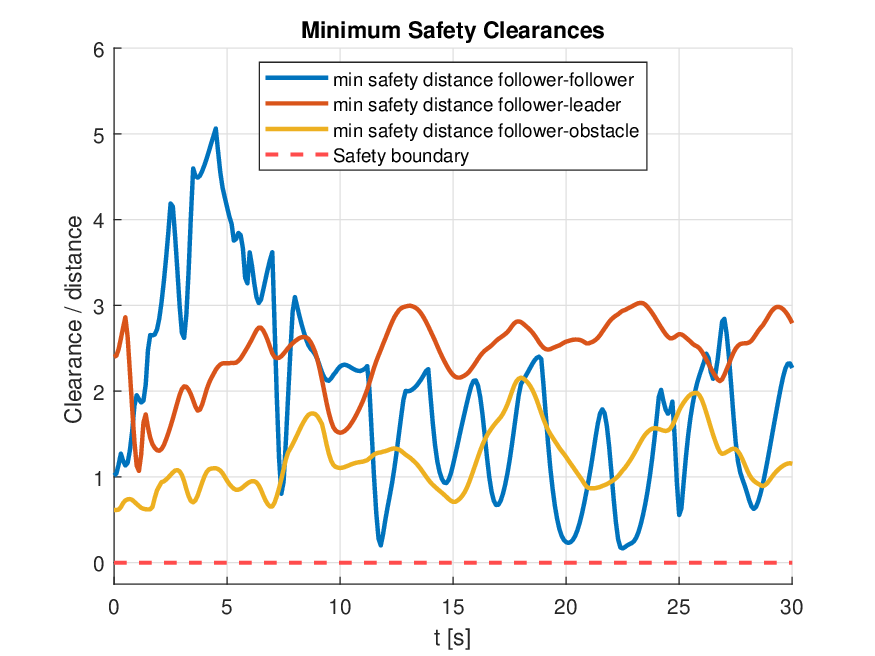}
    \caption{Minimum clearance from \cite{koulong2025wc}}
    \label{fig:2x}
    \end{subfigure}
    \hfill
    \begin{subfigure}[b]{0.410\textwidth}
    \centering
    \includegraphics[width=2.9in]{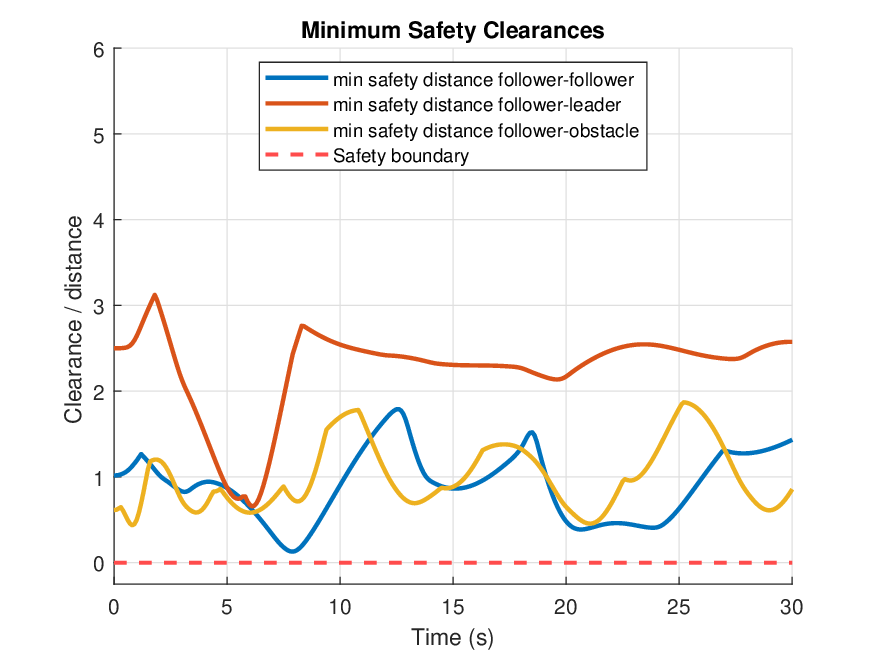}
    \caption{Minimum clearance}
    \label{fig:3x}
    \end{subfigure}
    \caption{Minimum Safety Clearance} 
    \label{fig:combined0x}
\end{figure}

\scalebox{.64}{
\small
\setlength{\tabcolsep}{6pt}
\renewcommand{\arraystretch}{1.1} 
\hspace{-19pt}    \begin{tabular}{|c|c|}
\hline 
Agent  & $\dot{x}_{1,k}^{i} = x_{2,k}^{i}$,
$\dot{x}_{2,k}^{i} = x_{3,k}^{i}$, $k=1,2$, with $\dot{x}_{3,k}^{i} = f^i_k(x^i) + u^i_k + w^i_k$ presented below:  \tabularnewline
\hline 
\hline 
Leader & $\begin{aligned}
    f^0_1(x^0) &= -\,5\,\bigl(x_{3,1}^{0}\;+\;0.36\,x_{1,1}^{0}\;+\;0.84\,x_{2,1}^{0}\; +\; 0.15\,(x_{1,1}^{0})^3\; -\;0.4\,\sin\bigl(0.8\,t)\bigr),\\
f^0_2(x^0) &= -\,5\,\bigl(x_{3,2}^{0}\;+\;0.36\,x_{1,2}^{0}\;+\;0.84\,x_{2,2}^{0}\; \; 0.15\,(x_{1,2}^{0})^3\; -\;0.4\,\sin\bigl(0.8\,t + \frac{\pi}{2})\bigr),
\\
w_1^3(t) &= 0.20\,\sin(0.9\,t), \, w_2^3(t) = 0.25\,\sin\!\big(1.1\,t-\tfrac{\pi}{7}\big);\,\, u^0_1=u^0_2=0
\end{aligned}$ \tabularnewline
\hline 
Agent 1 & $\begin{aligned}
    f^1_1(x^1) = -k_a^1\!\big(x_{3,1}^1 + 0.49\,x_{1,1}^1 + 1.12\,x_{2,1}^1 + 0.12\,(x_{1,1}^1)^3   - 0.25\,\tanh(0.6\,x_{1,1}^1)\big), \\
    f^1_2(x^1) = -k_a^1\!\big(x_{3,2}^1 + 0.49\,x_{1,2}^1 + 1.12\,x_{2,2}^1 + 0.12\,(x_{1,2}^1)^3   - 0.25\,\tanh(0.6\,x_{1,2}^1)\big), \\
    w_1^1(t) = 0.20\,\sin(0.9\,t), \, w_2^1(t) = 0.15\,\sin\!\big(1.1\,t+\tfrac{\pi}{7}\big).

\end{aligned}$ \tabularnewline
\hline 
Agent 2 & $\begin{aligned}
    f^2_1(x^2) = -k_a^2\!\big(x_{3,1}^2 +0.36\,x_{1,1}^2 + 0.84\,x_{2,1}^2 + 0.18\,(x_{1,1}^2)^3  - 0.15\,(x_{1,2}^2)^2\,\tanh(x_{1,1}^2)\big), \\
f^2_2(x^2) = -k_a^2\!\big(x_{3,2}^2+0.36\,x_{1,2}^2 + 0.84\,x_{2,2}^2 + 0.18\,(x_{1,2}^2)^3   - 0.15\,(x_{1,1}^2)^2\,\tanh(x_{1,2}^2)\big), 
\\
w_1^2(t) = 0.18\,\sin(1.1\,t+0.3), \, w_2^2(t) = 0.18\,\sin\!\big(0.8\,t-\tfrac{\pi}{5}\big).
\end{aligned}$ \tabularnewline
\hline 
Agent 3 & $\begin{aligned}
    f^3_1(x^3) = -k_a^3\!\big(x_{3,1}^3+0.25\,\tanh(x_{1,1}^3) + 0.9\,\tanh(x_{2,1}^3) - 0.20\,\sin(0.7\,t)\big), \\
f^3_2(x^3) = -k_a^3\!\big(x_{3,2}^3+0.25\,\tanh(x_{1,2}^3) + 0.9\,\tanh(x_{2,2}^3)  - 0.20\,\sin(0.9\,t+\tfrac{\pi}{3})\big), 
\\
w_1^3(t) = 0.18\,\sin(1.1\,t+0.3), \, w_2^3(t) = 0.18\,\sin\!\big(0.8\,t-\tfrac{\pi}{5}\big).
\end{aligned}$ \tabularnewline
\hline 
Agent 4 & $\begin{aligned}
f^4_1(x^4) = -k_a^4\!\big(x_{3,1}^4+0.4225\,x_{1,1}^4 + 0.975\,x_{2,1}^4   + 0.08\,\big(x_{1,1}^4+x_{1,2}^4\big)^3\big), 
\\
f^4_2(x^4) = -k_a^4\!\big(x_{3,2}^4+0.4225\,x_{1,2}^4 + 0.975\,x_{2,2}^4 - 0.08\,\big(x_{1,1}^4-x_{1,2}^4\big)^3\big), \\
w_1^4(t) = 0.12\,\sin(0.6t), \, w_2^4(t) = 0.12 \sin\!\big(0.6t+\tfrac{\pi}{2}\big).
\end{aligned}$ \tabularnewline
\hline 
Agent 5 & $\begin{aligned}
f^5_1(x^5) = -k_a^5\!\big(x_{3,1}^5+0.3025\,x_{1,1}^5 + 0.88\,x_{2,1}^5 + 0.15\,(x_{1,1}^5)^3    - 0.20\,\tanh(0.5\,x_{1,1}^5)\big), \\
f^5_2(x^5) = -k_a^5\!\big(x_{3,2}^5+0.3025\,x_{1,2}^5 + 0.88\,x_{2,2}^5 + 0.15\,(x_{1,2}^5)^3    - 0.20\,\tanh(0.5\,x_{1,2}^5)\big), \\
w_1^5(t) = 0.22\,\sin\!\big(0.9\,t+\tfrac{\pi}{8}\big), \, w_2^5(t) = 0.20\,\sin(1.0\,t).
\end{aligned}$  \tabularnewline
\hline 
\end{tabular}
}

\subsection*{Results Analysis}

Both methods enforce safety by tightening nominal eCBF constraints via an RPI tube; however, larger tubes shrink the nominal feasible set, forcing the planner to expend control authority on robustness rather than formation regulation.  Because the baseline \cite{koulong2025wc} provides a general-purpose safety framework, it must conservatively bound complex internal dynamics, yielding a large, restrictive tube. This explains why Figure~\ref{fig:4}, simulated under these general-purpose baseline constraints, exhibits large, persistent formation-offset errors, and why Figure~\ref{fig:2x} exhibits excessively large safety margins. The safety layer consumes a disproportionate share of the control authority, preventing the nominal formation objective from being closely recovered. 

By contrast, exploiting knowledge of dynamics yields far less conservative bounds, as evidenced by Figure~\ref{fig:combined0}, where both error components are driven much closer to zero, yielding bounded residuals that are significantly smaller than the baseline. By directly canceling undesired internal nonlinearities via feed-forward compensation rather than bounding them as general disturbances, this design fundamentally shrinks the deviation dynamics seen by the safety mechanism. Consequently, the resulting RPI tubes and tightened safe sets are far less restrictive, allowing the controller to efficiently recover the prescribed leader-relative offsets without perpetually fighting its own robustness margins.

This reduction in conservatism is further corroborated by the clearance comparison in Figure~\ref{fig:combined0x}. Figure~\ref{fig:2x}, utilizing the baseline method of \cite{koulong2025wc}, shows relatively large minimum clearances—especially for follower–follower distances—which is a direct byproduct of a highly conservative safety layer. In Figure~\ref{fig:3x}, the current work keeps all minimum clearances strictly positive, ensuring formal safety is preserved, but the enforced margins are noticeably smaller. This is the exact signature of a controller operating with mathematically tighter, less conservative safety bounds rather than weaker safety. Because the feed-forward mechanism shrinks the required safety buffer, the trajectories in Figure~\ref{fig:3} are allowed to stay much closer to the desired formation manifold. Conversely, the trajectories in Figure~\ref{fig:2} remain visibly distorted by the intrusive robust safety correction. 


\section{Conclusion}\label{sec:CONCLUSION}

Explicit, distributed safety bounds have been established for anticipative leader-follower tracking in high-order Brunovsky nonlinear multi-agent systems subject to bounded disturbances. By deriving robust forward-invariance certificates for a feedforward-augmented ancillary control policy, the proposed framework exploits the resulting deviation dynamics to entirely circumvent restrictive Lipschitz-dependent feasibility conditions. This yields explicit, closed-form tube radii that fundamentally shrink the required exponential control barrier function (eCBF) tightening margins, minimizing the conflict between robust safety enforcement and formation tracking while preserving rigorous guarantees. Furthermore, leveraging a network-based synchronization error formulation, closed-form global formation error bounds were derived that connect individual local safety tubes to collective multi-agent performance via the minimum singular value of the augmented communication graph. Future work will address switching topologies via either autonomous (i.e., event-triggered) communication switchings or controlled switchings (e.g., dwell-time conditions), as well as the consideration of communication delays, and implementations and experimental validations on multi-robot platforms.

\bibliographystyle{ieeetr}
\bibliography{Bibliographies}

@article{Oh2015UAVSurvey,
title = {A survey of multi-agent formation control},
journal = {Automatica},
volume = {53},
pages = {424-440},
year = {2015},
issn = {0005-1098},
doi = {https://doi.org/10.1016/j.automatica.2014.10.022},
url = {https://www.sciencedirect.com/science/article/pii/S0005109814004038},
author = {Kwang-Kyo Oh and Myoung-Chul Park and Hyo-Sung Ahn}
}

@article{Olfati-Saber2007ConsensusCoopControl,
  title={Consensus and Cooperation in Networked Multi-Agent Systems},
  author={Reza Olfati-Saber and J. Alex Fax and Richard M. Murray},
  journal={Proceedings of the IEEE},
  year={2007},
  volume={95},
  pages={215-233},
  url={https://api.semanticscholar.org/CorpusID:6533249}
}

@INPROCEEDINGS{Chen2016FormationRoboticsSurvey,
  author={Yang Quan Chen and Zhongmin Wang},
  booktitle={International Conference on Intelligent Robots and Systems}, 
  title={Formation control: a review and a new consideration}, 
  year={2005},
  volume={},
  number={},
  pages={3181-3186},
  doi={10.1109/IROS.2005.1545539}}

@article{DunbarMurray2006,
title = {Distributed receding horizon control for multi-vehicle formation stabilization},
journal = {Automatica},
volume = {42},
number = {4},
pages = {549-558},
year = {2006},
issn = {0005-1098},
doi = {https://doi.org/10.1016/j.automatica.2005.12.008},
url = {https://www.sciencedirect.com/science/article/pii/S0005109806000136},
author = {William B. Dunbar and Richard M. Murray}
}

@article{Dai2017DMPCFormationSurveyLike,
title = {Distributed MPC for formation of multi-agent systems with collision avoidance and obstacle avoidance},
journal = {Journal of the Franklin Institute},
volume = {354},
number = {4},
pages = {2068-2085},
year = {2017},
issn = {0016-0032},
doi = {https://doi.org/10.1016/j.jfranklin.2016.12.021},
url = {https://www.sciencedirect.com/science/article/pii/S0016003216304926},
author = {Li Dai and Qun Cao and Yuanqing Xia and Yulong Gao}
}

@article{Mayne2005RMPC,
title = {Robust model predictive control of constrained linear systems with bounded disturbances},
journal = {Automatica},
volume = {41},
number = {2},
pages = {219-224},
year = {2005},
issn = {0005-1098},
doi = {https://doi.org/10.1016/j.automatica.2004.08.019},
url = {https://www.sciencedirect.com/science/article/pii/S0005109804002870},
author = {D.Q. Mayne and M.M. Seron and S.V. Raković}
}

@article{XuAmes2015RobustCBF,
title = {Robustness of Control Barrier Functions for Safety Critical Control**This work is partially supported by the National Science Foundation Grants 1239055, 1239037 and 1239085.},
journal = {IFAC-PapersOnLine},
volume = {48},
number = {27},
pages = {54-61},
year = {2015},
note = {Analysis and Design of Hybrid Systems ADHS},
issn = {2405-8963},
doi = {https://doi.org/10.1016/j.ifacol.2015.11.152},
url = {https://www.sciencedirect.com/science/article/pii/S2405896315024106},
author = {Xiangru Xu and Paulo Tabuada and Jessy W. Grizzle and Aaron D. Ames}
}

@article{Jankovic2018RobustCBF,
title = {Robust control barrier functions for constrained stabilization of nonlinear systems},
journal = {Automatica},
volume = {96},
pages = {359-367},
year = {2018},
issn = {0005-1098},
doi = {https://doi.org/10.1016/j.automatica.2018.07.004},
url = {https://www.sciencedirect.com/science/article/pii/S0005109818303509},
author = {Mrdjan Jankovic}
}

@misc{Liu2023IterativeMPCDHOCBF,
      title={Iterative Convex Optimization for Model Predictive Control with Discrete-Time High-Order Control Barrier Functions}, 
      author={Shuo Liu and Jun Zeng and Koushil Sreenath and Calin A. Belta},
      year={2023},
      eprint={2210.04361},
      archivePrefix={arXiv},
      primaryClass={math.OC},
      url={https://arxiv.org/abs/2210.04361}, 
}

@INPROCEEDINGS{Zeng2021MPCDiscreteCBF,
  author={Zeng, Jun and Zhang, Bike and Sreenath, Koushil},
  booktitle={American Control Conference (ACC)}, 
  title={Safety-Critical Model Predictive Control with Discrete-Time Control Barrier Function}, 
  year={2021},
  volume={},
  number={},
  pages={3882-3889},
  doi={10.23919/ACC50511.2021.9483029}}

@INPROCEEDINGS{XiaoBelta2019HOCBF,
  author={Xiao, Wei and Belta, Calin},
  booktitle={IEEE 58th Conference on Decision and Control (CDC)}, 
  title={Control Barrier Functions for Systems with High Relative Degree}, 
  year={2019},
  volume={},
  number={},
  pages={474-479},
  doi={10.1109/CDC40024.2019.9029455}}

@INPROCEEDINGS{NguyenSreenath2016ECBF,
  author={Nguyen, Quan and Sreenath, Koushil},
  booktitle={American Control Conference (ACC)}, 
  title={Exponential Control Barrier Functions for enforcing high relative-degree safety-critical constraints}, 
  year={2016},
  volume={},
  number={},
  pages={322-328},
  doi={10.1109/ACC.2016.7524935}}

@article{Keviczky2008,
  author       = {Tam{\'{a}}s Keviczky and
                  Francesco Borrelli and
                  Kingsley Fregene and
                  Datta N. Godbole and
                  Gary J. Balas},
  title        = {Decentralized Receding Horizon Control and Coordination of Autonomous
                  Vehicle Formations},
  journal      = {{IEEE} Trans. Control. Syst. Technol.},
  volume       = {16},
  number       = {1},
  pages        = {19--33},
  year         = {2008},
  url          = {https://doi.org/10.1109/TCST.2007.903066},
  doi          = {10.1109/TCST.2007.903066},
  bibsource    = {dblp computer science bibliography, https://dblp.org}
}

@ARTICLE{Li2012,
  author={Li, Jianzhen and Ren, Wei and Xu, Shengyuan},
  journal={IEEE Transactions on Automatic Control}, 
  title={Distributed Containment Control with Multiple Dynamic Leaders for Double-Integrator Dynamics Using Only Position Measurements}, 
  year={2012},
  volume={57},
  number={6},
  pages={1553-1559},
  doi={10.1109/TAC.2011.2174680}}

@ARTICLE{Shorinwa2024,
  author={Shorinwa, Ola and Schwager, Mac},
  journal={IEEE Transactions on Automatic Control}, 
  title={Distributed Model Predictive Control via Separable Optimization in Multiagent Networks}, 
  year={2024},
  volume={69},
  number={1},
  pages={230-245},
  doi={10.1109/TAC.2023.3269338}}

@article{Fu70088,
author = {Fu, Chenlong and Wu, Jinxian and Dai, Li and Xia, Yuanqing},
title = {Distributed MPC-Based Trajectory Tracking Control for a Multi-Quadrotor UAV Slung Load System},
journal = {IET Control Theory \& Applications},
volume = {19},
number = {1},
pages = {e70088},
doi = {https://doi.org/10.1049/cth2.70088},
url = {https://ietresearch.onlinelibrary.wiley.com/doi/abs/10.1049/cth2.70088},
eprint = {https://ietresearch.onlinelibrary.wiley.com/doi/pdf/10.1049/cth2.70088},
year = {2025}
}

@article{FARINA20121088,
title = {Distributed predictive control: A non-cooperative algorithm with neighbor-to-neighbor communication for linear systems},
journal = {Automatica},
volume = {48},
number = {6},
pages = {1088-1096},
year = {2012},
issn = {0005-1098},
doi = {https://doi.org/10.1016/j.automatica.2012.03.020},
url = {https://www.sciencedirect.com/science/article/pii/S000510981200129X},
author = {Marcello Farina and Riccardo Scattolini}
}

@article{Kan2016DirectedRandomGraphs,
title = {Leader–follower containment control over directed random graphs},
journal = {Automatica},
volume = {66},
pages = {56-62},
year = {2016},
issn = {0005-1098},
doi = {https://doi.org/10.1016/j.automatica.2015.12.016},
url = {https://www.sciencedirect.com/science/article/pii/S0005109815005488},
author = {Zhen Kan and John M. Shea and Warren E. Dixon}
}

@article{RIVERSO20142179,
title = {Plug-and-play model predictive control based on robust control invariant sets},
journal = {Automatica},
volume = {50},
number = {8},
pages = {2179-2186},
year = {2014},
issn = {0005-1098},
doi = {https://doi.org/10.1016/j.automatica.2014.06.004},
url = {https://www.sciencedirect.com/science/article/pii/S0005109814002325},
author = {Stefano Riverso and Marcello Farina and Giancarlo Ferrari-Trecate}
}

@article{Riverso2012PlugandPlayDM,
  title={Plug-and-Play decentralized Model Predictive Control},
  author={Stefano Riverso and Marcello Farina and Giancarlo Ferrari-Trecate},
  journal={2012 IEEE 51st IEEE Conference on Decision and Control (CDC)},
  year={2012},
  pages={4193-4198},
  url={https://api.semanticscholar.org/CorpusID:8226347}
}

@misc{wang2025distributedsafetycriticalmpcmultiagent,
      title={Distributed Safety-Critical MPC for Multi-Agent Formation Control and Obstacle Avoidance}, 
      author={Chao Wang and Shuyuan Zhang and Lei Wang},
      year={2025},
      eprint={2508.19678},
      archivePrefix={arXiv},
      primaryClass={eess.SY},
      url={https://arxiv.org/abs/2508.19678}, 
}

@Inbook{RenBeardMcLain2003,
author="Ren, Wei
and Beard, Randal W.
and McLain, Timothy W.",
editor="",
title="Coordination Variables and Consensus Building in Multiple Vehicle Systems",
bookTitle="Cooperative Control: A Post-Workshop Volume 2003 Block Island Workshop on Cooperative Control",
year="2005",
publisher="Springer Berlin Heidelberg",
address="Berlin, Heidelberg",
pages="171--188",
isbn="978-3-540-31595-7",
doi="10.1007/978-3-540-31595-7_10",
url="https://doi.org/10.1007/978-3-540-31595-7_10"
}

@ARTICLE{LiRenXuTAC2012_PositionOnlyContainment,
  author={Li, Jianzhen and Ren, Wei and Xu, Shengyuan},
  journal={IEEE Transactions on Automatic Control}, 
  title={Distributed Containment Control with Multiple Dynamic Leaders for Double-Integrator Dynamics Using Only Position Measurements}, 
  year={2012},
  volume={57},
  number={6},
  pages={1553-1559},
  doi={10.1109/TAC.2011.2174680}}

@INPROCEEDINGS{LiRenXuACC2011_PositionOnlyTracking,
  author={Li, Jianzhen and Ren, Wei and Xu, Shengyuan},
  booktitle={Proceedings of the 2011 American Control Conference}, 
  title={Distributed coordinated tracking with multiple dynamic leaders for double-integrator agents using only position measurements}, 
  year={2011},
  volume={},
  number={},
  pages={2192-2197},
  doi={10.1109/ACC.2011.5991402}}

@INPROCEEDINGS{koulong2025acc,
  author={Koulong, Armel and Pakniyat, Ali},
  booktitle={American Control Conference (ACC)}, 
  title={Distributed Adaptive Consensus with Obstacle and Collision Avoidance for Networks of Heterogeneous Multi-Agent Systems}, 
  year={2025},
  volume={},
  number={},
  pages={3602-3607},
  doi={10.23919/ACC63710.2025.11107682}}

@article{koulong2025wc,
      title={Robust Multi-Agent Safety via Tube-Based Tightened Exponential Barrier Functions}, 
      author={Armel Koulong and Ali Pakniyat},
      year={2026},
      journal={(accepted for publication in) Nonlinear Analysis: Hybrid Systems; Available: arXiv preprint arXiv:2510.22514},
      eprint={2510.22514},
      archivePrefix={arXiv},
      primaryClass={eess.SY},
      url={https://arxiv.org/abs/2510.22514}, 
}

\end{document}